\documentclass[a4paper,12pt]{amsart}
\usepackage{graphicx}
\usepackage{color}
\usepackage{amssymb}
\usepackage{graphicx}
\usepackage{epstopdf}
\usepackage{listings}
\usepackage{amsmath}
\usepackage{amsfonts}
\usepackage{mathtools}
\usepackage{tikz} % To generate the plot from csv
\usepackage{pgfplots}
\usepackage{geometry}
\usepackage{setspace}
\usepackage{caption}
\usepackage{subcaption}
\usepackage{lipsum}
\usepackage{wrapfig}
\usepackage{float}
\usepackage{todonotes}
\usepackage{bm,nicefrac}

\usepackage{listings}

\usepackage{extarrows}
\usepackage{enumitem}
\setlist[enumerate]{itemsep=0mm}
\usepackage{color}
\newcommand{\blue}{\textcolor{black}}

\newcommand{\EE}{\mathbb E}
\newcommand{\II}{\mathbb I}
\newcommand{\PP}{\mathbb P}

\newcommand{\Q}{\mathcal Q}

\newcommand{\cD}{\mathcal D}
\newcommand{\R}{\mathcal R}

\newcommand{\C}{\mathcal C}

\newcommand{\G}{\mathcal G}

\newcommand{\B}{\mathcal B}

\newtheorem{theorem}{Theorem}

\newtheorem{corollary}{Corollary}

\begin{document}

\title[Expected number of uninhibited  RAF sets] {The expected number of viable autocatalytic sets in \blue{chemical reaction systems}}
\author{Stuart Kauffman and Mike Steel}

\bigskip

\begin{abstract}
The emergence of self-sustaining autocatalytic networks in chemical reaction systems has been studied as a possible mechanism for modelling how
living systems first arose.  It has been known  for several decades that such networks will form within systems of polymers (under cleavage and ligation reactions)  under a simple process of random catalysis, and  this process has since been mathematically analysed. 
In this paper, we provide an exact expression for the expected number of self-sustaining autocatalytic networks that will form in a general chemical reaction system, and the expected number of these networks that will also be uninhibited (by some molecule produced by the system).  Using these equations, we are able to describe the patterns of catalysis and inhibition that maximise or minimise the expected number of such networks. We apply our results to derive a general theorem concerning the trade-off between catalysis and inhibition, and to provide some insight into the extent to which the expected number of self-sustaining autocatalytic networks coincides with the probability that at least one such system is present.

\end{abstract}

\maketitle

\noindent {\em Address:}
                 M. Steel (corresponding author):\\
              Biomathematics Research Centre, \\University of Canterbury, Christchurch, New Zealand \\
              Tel.: +64-33667001\\
              \email{mike.steel@canterbury.ac.nz}  
              \\
              \\
             Stuart Kauffman\\
             Affiliate Professor Institute for Systems Biology \\
             Emeritus Professor, Biochemistry and Biophysics, University of Pennsylavania, PA, USA\\
              \email{stukauffman@gmail.com}  
              \\
              \\


\newpage

\section{Introduction}

A key step in the origin of life is the formation of a metabolic network that is both self-sustaining and collectively autocatalytic \cite{ars, hay, liu, vai, vas, xav}. Systems that combine these two general properties have been studied  within a formal framework that is sometimes referred to as  RAF theory \cite{hor17}. 
We give precise definitions shortly but, roughly speaking, a `RAF'  \blue{(=Reflexively Autocatalytic and F-generated)} set is a subset of reactions where the reactants and at least one catalyst of each reaction in the subset can be produced from an available food set by using reactions from within the subset only.

The study of RAFs traces back to pioneering work on `collectively autocatalytic sets'  in polymer models of early life \cite{kau71, kau86}, which was subsequently developed  mathematically (see \cite{hor19, hor17} and the references there-in). RAF algorithms have been applied recently to investigate the traces of earliest metabolism that can be detected in large metabolic databases  across bacteria and archaea \cite{xav}, leading to the development of an open-source program to analyse and visualise RAFs in complex biochemical systems \cite{cat}.  RAF theory  overlaps with other graph-theoretic approaches in which the emergence of directed cycles in reaction graphs plays a key role \cite{bol, j1, j2}, and  is also related to (M, R) systems \cite{cor, jar10} and chemical organisation theory \cite{dit2}.

 RAF theory has  also been applied in other fields, including ecology \cite{caz18} and cognition \cite{gab17}, and the ideas may have application in other contexts. In economics, for instance,  the production of consumer items can be viewed as a catalysed reaction; for example,  the production of a wooden table involves nails and wood (reactants) and a hammer (a catalyst, as it is not used up in the reaction but makes the reaction happen much more efficiently) and the output (reaction product) is the table. On a larger scale, a factory is  a catalyst for the production of the items produced in it from reactants brought into the factory.  In both these examples, notice that each reactant may either be a raw material (i.e. the elements of a `food set') or a  products of other (catalysed) reactions, whereas the products may, in turn, be reactants, or catalysts,  for other catalysed reactions. Products can sometimes also {\em inhibit} reactions; for example, the production of internal combustion engines resulted in processes for building steam engines being abandoned. 

In this paper, we extend RAF theory further by investigating the impact of different modes of catalysis and inhibition on the appearance of (uninhibited) RAF subsets. We focus on the expected number of such sets (rather than on the probability that at least one such set exists which has been the focus of nearly all earlier RAF studies \cite{fil, mos}). \blue{Using a mathematical approach, we derive explicit  and exact analytical expressions for the expected number of such uninhibited RAF subsets,}  as well as providing some insight into the expected population sizes of RAFs for the catalysis  rate at which they first appear (as we discuss in Section~\ref{relation}).   \blue{In particular, we show that for simple systems, with an average catalysis rate that is set at  the level where RAFs first appear, the expected number of RAFs depends strongly on the variability of catalysis across molecules. At one extreme (uniform catalysis),  the expected number of RAFs is small (e.g. 1, or a few), while at the other extreme (all-or-nothing catalysis) the expected number of RAFs  grows exponentially with the size of the system.}

\blue{The motivation for looking at the expected number of RAFs (rather than the probability that a  RAF exists) is twofold. Firstly, by focusing on expected values it is possible to present certain exact results (in Theorem~\ref{thm1}), rather than just inequalities or asymptotic results, while  still gaining some information about the probability that a RAF exists.  Secondly, in origin of life studies, it is relevant to consider populations of self-sustaining autocatalytic chemical networks, which may be subject to competition and selection, a topic which has explored by others (see e.g. \cite{sza, vas, vir}), and information concerning the likely diversity of RAFs available in a given chemical reaction system is therefore a natural question.  In previous analyses where RAFs have been identified, subsequent analysis has revealed a large number of RAFs present within the RAF; for example, for a 7-reaction RAF  in a laboratory-based study involving RNA-ribosymes (from \cite{vai}) more than half of the $2^7 = 128$ subsets of this RAF are also RAFs ({\em cf.} Fig. 5 of  \cite{ste}).  Simulation studies involving Kauffman's binary polymer model have also identified a large number of RAFs present once catalysis rises above the level at which RAFs first appear \cite{hor15}. }

\blue{The structure of this paper is as follows. We begin with some formal definitions, and then described different models for catalysis and inhibition. In Section~\ref{gensec}, we present the main mathematical result, along with some remarks, and  proof.  We then present a number of consequences of our main result, beginning with a generic result concerning  the impact of inhibition when catalysis is uniform.  We then investigate  the impact of different catalysis distributions on the expected number of RAF arising in `elementary' chemical reaction systems, focusing on the catalysis rate at which RAFs first appear.  We end with some brief concluding comments.}

\subsection{Definitions}

Let $X$ be a set of molecule types; $R$ a set of reactions, where each reaction consists of a subset of molecule types as input (`reactants') and a set of molecule types as 
outputs (`products'); and let $F$ be a subset of $X$ (called a `food set'). We refer to the triple $\Q=(X, R, F)$ as a {\em chemical reaction system with food set} and, unless stated otherwise, we impose  no  further restrictions on $\Q$ (e.g. it need not correspond to a system of polymers and a reaction can have any positive number of reactants and any positive number of products).
 
 Given a reaction $r \in R$, we let $\rho(r) \subseteq X$ denote the set of reactants of $r$ and $\pi(r)$ denote the set of products of $r$. 
 Moreover, given a subset $R'$ of $R$, we let $\pi(R') = \bigcup_{r \in R'} \pi(r).$
 
 A subset $R'$ of $R$ is said to be {\em $F$-generated} if $R'$ can be ordered $r_1, r_2, \ldots, r_{|R'|}$ so that 
 $\rho(r_1) \subseteq F$ and for each $i \in \{2, \ldots, |R|\}$, we have $\rho(r_i) \subseteq F \cup \pi(\{r_1, \ldots, r_{i-1}\})$. In other words, $R'$ is $F$-generated if the $R'$ can be built up by starting from one reaction that has all its reactants in the food set, then adding reactions in such a way that each added reaction has each of its reactants present either in the food set or as a product of a reaction in the set generated so far.

Now suppose that certain molecule types in $X$ can catalyse certain reactions in $R$.
A subset $R'$ of $R$ is said to be {\em Reflexively Autocatalytic and F-generated} (more briefly, a {\em RAF}) if $R'$ is nonempty and each reaction $r \in R'$ is catalysed by
at least one molecule type in $F \cup \pi(R')$ and $R'$ is $F$-generated.

We may also allow certain molecule types to also inhibit reactions in $R$, in which case a subset 
 $R'$ of $R$ is said to be an {\em uninhibited RAF} (uRAF) if 
$R'$ is a RAF and no reaction in $R'$ is inhibited by any molecule type in $F \cup \pi(R')$.   \blue{The notion of a uRAF was first defined and studied in \cite{mos}.}
\blue{Notice that inhibition is being applied in a strong sense:  a reaction $r$ cannot be part of a uRAF if $r$  is inhibited by at least one molecule type present, regardless of how many  molecule types are catalysts  for $r$ and present in the uRAF}.

Since a union of RAFs is also a RAF, when a RAF exists in a system, there is a unique maximal RAF. However, the same does not apply to uRAFs -- in particular, the union of two uRAFs can fail to be a uRAF.  
These concepts are illustrated in Fig.~\ref{fig1}.
 \begin{figure}[h]
\centering
\includegraphics[scale=1.1]{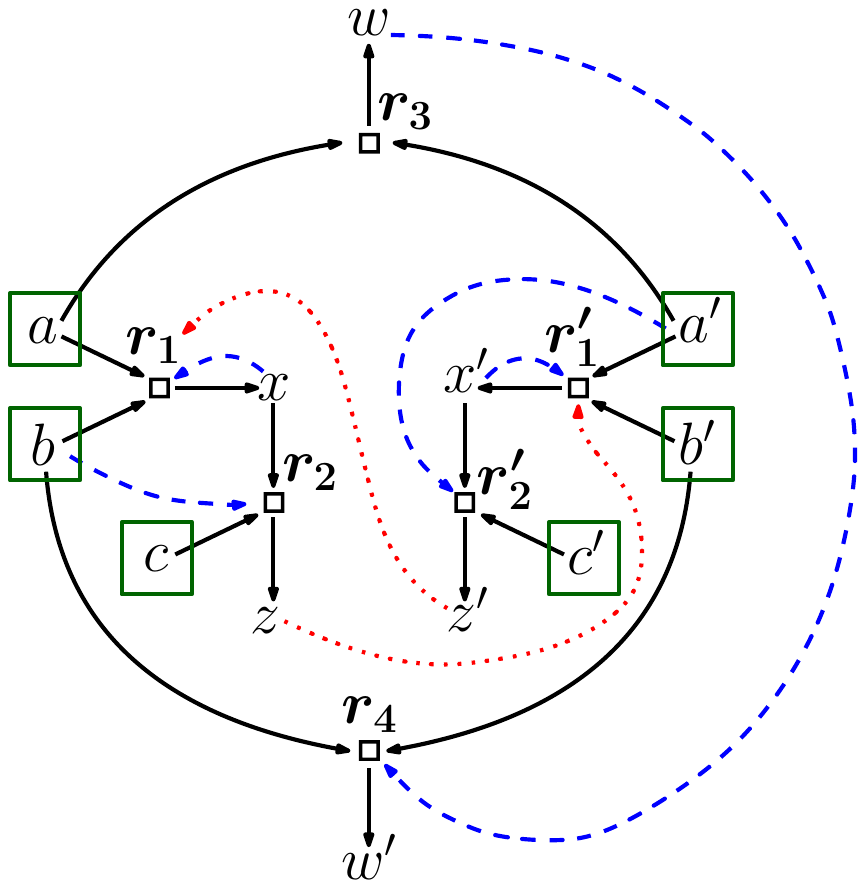}
\caption{A chemical reaction system consisting of  the set of molecule types $X=\{a, b, c, a', b', c', x, x', w,w', z,z'\}$, a food set $F=\{a, b, c, a', b', c'\}$ \blue{(each placed inside a green box)}  and the reaction set 
$R=\{r_1,  r_2,  r_1', r_2', r_3, r_4\}$ \blue{(bold, beside small white-filled squares)}.  Solid arcs indicate two reactants entering a reaction and a product coming out.  
Catalysis is indicated by dashed arcs (blue) and inhibition (also called blocking) is indicated by dotted arcs (red).   The full set of
reactions is not a RAF, but it contains several RAFs that are contained in the unique maximal RAF $R'=\{r_1, r_1', r_2, r_2'\}$ (note that $r_4$ is not part of this RAF even though it is catalysed and the reactants of $r_4$ are present in the food set). 
 The maximal RAF $R'$ is not a uRAF (e.g. $r'_1$ is inhibited by $z$ which is a product of $r_2$); however, $\{r_1, r_2\}$ and $\{r_1', r_2'\}$ are uRAFs, and so are $\{r_1\}, \{r_1'\}$ and $\{r_1, r_1'\}$.  }
\label{fig1}
\end{figure}

 \section{Modelling catalysis and inhibition} 
We will model catalysis and also blocking (inhibition) by random processes.
To provide for greater generality, we allow the possibility that elements in a subset $C^{-}$ (respectively, $B^{-}$) of the  food set cannot catalyse (respectively block) any reaction in $R$.  Let $c=|F \setminus C^{-}|$ and $b=|F \setminus B^{-}|$. Thus $c$  (respectively $b$) is the number of food elements that are possible catalysts (respectively blockers). 

Suppose that each molecule type $x \in X\setminus C^{-}$ has an associated probability $C_x$ of catalysing any given reaction in $R$. \blue{The values $C_x$ are sampled independently from a distribution $\cD$, for each $x \in X$.}
  This results in a random assignment of catalysis  (i.e. a random subset $\chi$ of $X \times \R$), where $(x,r)  \in \chi$ if $x$ catalyses $r$. Let
$\C_{x,r}$ be the event that $x$ catalyses $r$.

We assume that:
\begin{itemize} 
\item[($I_1$)]  $\C=(C_x, x\in X\setminus C^{-})$ is a collection of independent random variables.
\item[($I_2$)]  Conditional on $\C$, $(\C_{x,r}:  x\in X \setminus C^-, r \in R)$ is a collection of independent events.
\end{itemize}
 Since the distribution of $C_x$ is the same for all $x \in X\setminus C^-$, we \blue{will use $C$ to denote an arbitrary random variable sampled from the distribution $\cD$.} Let 
$\mu_C = \EE[C]$ and, for $i\geq 0$, let $\lambda_i$ be the $i$--th moment of $1-C$; that is:
$$\lambda_i =\EE[(1-C)^i].$$
Although our results concern general catalysis distributions, we will pay particular attention to three forms of catalysis \blue{which have been considered in previous studies (e.g. \cite{hor16}), and which will be compared in our analyses.}
\begin{itemize}
\item The  {\em uniform model:}  Each $x \in X\setminus C^-$ catalyses each reaction in $\R$ with a fixed probability $p$. Thus, $C =p$ with probability 1, and so $\mu_C = p$.  
\item The   {\em sparse model:}  $C= u$ with probability $\pi$ and $C =0$ with probability $1-\pi$, and so $\mu_C = u \pi$.
\item
The  {\em all-or-nothing model:} $C=1$ with probability $\pi$ and $C=0$ with probability $1-\pi$, and so $\mu_C = \pi$.
\end{itemize}

The uniform model is from Kauffman's binary polymer network and has been the default for most recent studies involving polymer models \cite{hor17}. 
 More  realistic catalysis scenarios  can be modelled by allowing $C$ to take a range of values 
values around $\mu_C$ with different probabilities. The   {\em sparse model}  generalises the uniform model slightly by allowing a (random) subset of molecule types to be catalysts. In this model, $\pi$ would typically be very small in applications (i.e.  most molecules are not catalysts but those few that are will catalyse a lot or reactions, as in the recent study of metabolic  \blue{origins, described in} \cite{xav}).  The all-or-nothing model is a  special case of the sparse model. 
The emergence of RAFs in these models (and others, including a power-law distribution) was  investigated in  \cite{hor16}. 

For these three models, the associated $\lambda_i$ values are given as follows: $\lambda_0=1$, and 
 for all $i\geq 1$:
\begin{equation}
\label{mu-eq}
\lambda_i = \begin{cases}
(1-\mu_C)^i, & \mbox{(uniform model)};\\
1-\pi + \pi(1-u)^i, & \mbox{(sparse model)};\\
1-\mu_C, & \mbox{(all-or-nothing model)}.
\end{cases}
\end{equation}

 In addition to catalysis, we may also allow random  blocking (inhibition) of reactions by molecules, formalised as follows. 
Suppose that each molecule type $x \in X\setminus B^{-}$ has an associated probability $B_x$ of blocking any given reaction in $R$. We will treat $B_x$ as a random variable taking values in  $[0,1]$ with a common distribution $\hat{\cD}$.  This results in a random assignment of blocking ( i.e. a random subset \blue{$\beta$} of $X \times \R$),  where \blue{$(x,r)  \in \beta$} if $x$ blocks reaction $r$. Let
$\B_{x,r}$ be the event that $x$ blocks $r$.  We assume that:
\begin{itemize} 
\item[($I'_1$)]  $\B=(B_x, x\in X\setminus B^{-})$ is a collection of independent random variables.
\item[($I'_2$)]  Conditional on $\B$, $(\B_{x,r}:  x\in X \setminus C^-, r \in R)$ is a collection of independent events.
\end{itemize}
 Since the distribution of $B_x$ is the same for all $x$, we will use $B$ to denote this random variable, let $\mu_B = \EE[B]$ and, for  $i\geq 0$, let: $$\hat{\lambda}_i =\EE[(1-B)^i].$$
 We also assume that catalysis and inhibition are independent of each other. Formally, this is the following condition:
\begin{itemize} 
\item[($I_3$)] $C$--random variables in ($I_1$, $I_2$) are independent of the $B$--random variables in ($I'_1$, $I'_2$).
\end{itemize}

 Note that $(I_3)$ allows the possibility that a molecule type $x$ both catalyses and blocks the same  reaction $r$ (the effect of this on uRAFs is the same as if $x$ just blocks $r$; (i.e. blocking  is assumed to trump catalysis)). 
 Notice also that $\lambda_0 = \hat{\lambda}_0 = 1$.
%In applications, it is assumed that $\mu_C$ and $\mu_B$ are both very close to 1 (i.e. each reaction catalyses or blocks few reactions), though this is not explicitly required in anything that follows. 

\section{Generic results}
\label{gensec}
To state our first result, we require two further definitions.  Let  $\mu_{\rm RAF}$ and $\mu_{\rm uRAF}$   denote the expected number of RAFs and uRAFs (respectively)  arising in $\Q$ under the random process of catalysis and inhibition described. 
 For  integers $k, s\geq1$ let $n_{k,s}$ be the number of F-generated subsets $R'$ of $R$ \blue{that have size $k$ and} for which the total number of non-food products in $X$ produced by reactions in $R'$ is $s$.  Note that $n_{k,s}=0$ for $s>\min\{|X|-F, k M\}$ where $M$ is the maximum number of products of any single reaction.

Part (i) of the following theorem gives an exact expression for $\mu_{\rm RAF}$ and $\mu_{\rm uRAF}$,  which we then use in Parts (ii) and (iii)  to describe the catalysis and inhibition distributions (having a given mean)  that minimise or maximise the expected number of RAFs and uRAFs.  We apply this theorem to particular systems in the next section.

\begin{theorem}
\label{thm1}
Let $\Q$ be any chemical reaction system with food set, accompanied by  catalysis and inhibition distributions $\cD$ and $\hat{\cD}$, respectively.
\begin{itemize}
\item[(i)]
The expected number of RAFs and uRAFs for $\Q$ is given as follows:

\begin{equation}
\label{mumu1}
\mu_{\rm RAF}= \sum_{k\geq 1,s\geq 0}  n_{k,s} \left(\sum_{i=0}^k (-1)^i \binom{k}{i} \lambda_i^{s+c}\right)
\end{equation}
and
\begin{equation}
\label{mumu2}
\mu_{\rm uRAF}= \sum_{k\geq 1, s\geq 0}  n_{k,s} \left(\sum_{i=0}^k (-1)^i \binom{k}{i} \lambda_i^{s+c}\right) \hat{\lambda}_k^{s+b}.
\end{equation}
\item[(ii)] 
Among all distributions $\cD$ on catalysis having a given mean $\mu_C$, the distribution that minimises the expected number of RAFs and uRAFs (for any  inhibition distribution) is the uniform model (i.e. $C = \mu_C$ with probability 1). 
\item[(iii)] Among all distributions $\hat{\cD}$ on inhibition having a given mean $\mu_B$, the following hold: 
\begin{itemize}
\item[(a)] the distribution that minimises the expected number of uRAFs (for any catalysis distribution) is the uniform model ($B = \mu_B$ with probability 1). \item[(b)] the distribution that maximises the expected number of uRAFs (for any catalysis distribution) is the all-or-nothing inhibition model (i.e. $B=1$ with probability $\mu_B$, and $B=0$ with probability $1-\mu_B)$.
\end{itemize}
\end{itemize}
\end{theorem}

\bigskip

\blue{We give the proof of Theorem~\ref{thm1} shortly, following some brief remarks.}

\subsection{Remarks}
\begin{itemize}

\item[(1)]
 If $P_{\rm RAF}$ and $P_{\rm uRAF}$ are the probability that $\Q$ contains a RAF and a uRAF, respectively, then these quantities are bounded above as follows:
$$P_{\rm RAF} \leq \mu_{\rm RAF} \mbox{ and } P_{\rm uRAF} \leq \mu_{\rm uRAF}.$$
This follows from the well-known inequality $\PP(V>0) \leq \EE[V]$ for any non-negative integer-valued random variable $V$, upon taking $V$
to be the number of RAFs (or the number of uRAFs).   We will explore the extent to which $P_{\rm RAF}$ underestimates $\mu_{\rm RAF}$ in Section~\ref{relation}.

\item[(2)]
Theorem~\ref{thm1} makes clear that the only relevant aspects of the network $(X, R)$ for  $\mu_{\rm RAF}$ and $\mu_{\rm uRAF}$ are encoded entirely within the coefficients  $n_{k,s}$ (the two stochastic terms depend only on $r$ and $s$ but not on further aspects of the network structure).  By contrast, an expression for the probabilities $P_{\rm RAF}$ and $P_{\rm uRAF}$  that a RAF or uRAF exists
requires more detailed information concerning the structure of the network. This is due to dependencies that arise in the analysis.
Notice also that Theorem~\ref{thm1} allows the computation of $\mu_{\rm uRAF}$ in $O(|R|^2 \times |X|)$ steps (assuming that the $\lambda_i, \hat{\lambda}_i$ and $n_{k,s}$ values are available). 

\item[(3)] 
Although the computation or estimation  of $n_{k,s}$ may be tricky in general systems,  Eqn.~(\ref{mumu1}) can still be useful (even with little or no information about $n_{k,s}$) for asking comparative  types of questions.  In particular, Parts (ii) and (iii)  provide results that are independent of the details of the network $(X, R, F)$. 
In particular,  Theorem~\ref{thm1}(ii) is consistent with simulation results in \cite{hor16} for Kauffman's binary polymer model, in which variable catalysis rates  (the sparse and all-or-nothing model) led to RAFs appearing at lower average catalysis values ($\mu_C$) than for uniform catalysis. 
 
\item[(4)]
\blue{For the uniform model, note that the term $\left(\sum_{i=0}^k (-1)^i \binom{k}{i} \lambda_i^{s+c}\right)$ in Eqns.~(\ref{mumu1}) and (\ref{mumu2}) simplifies to
$\left[ 1- (1-\mu_C)^{s+c}\right]^k$.}

\end{itemize}

\bigskip

\subsection{\blue{Proof of Theorem~\ref{thm1}}}

For Part (i), recall that  $\pi(R')$ denotes the set of products of reactions in $R'$.

\blue{ For $k, s \geq 1,$ let ${\rm FG}(k,s)$ denote the collection of subsets $R'$ of $R$ that satisfy all of the following three properties:
\begin{itemize}
\item[(i)] $R'$ has size $k$;
\item[(ii)] $R'$ is F-generated, and 
\item[(iii)] the number of non-food molecule types produced by reactions in $R'$ is $s$.
\end{itemize}
}
Thus, $$n_{k,s}= |{\rm FG}(k,s)|.$$ 
For $R' \subseteq  R$, let $\II_{R'}$ be the Bernoulli random variable 
that takes the value $1$ if each reaction in $R'$ is  catalysed by at least one product of a reaction in $R'$ or by an element of $F\setminus C^{-}$, and $0$ otherwise. 
Similarly, let $\hat{\II}_{R'}$  be the Bernoulli random variable 
that takes the value $1$ if no reaction in $R'$ is blocked by the product of any reaction in $R'$ or by an element of $F\setminus B^{-}$.   Then the random variable
$$\sum_{k\geq 1,s\geq 0}  \sum_{R' \in {\rm FG}(k,s)} \II_{\R'}\cdot \hat{\II}_{\R'}$$
counts the number of uRAFs  present, so we have:
\begin{equation}\label{nicer}
\mu_{\rm uRAF} = \EE\left[\sum_{k \geq 1, s\geq 0}  \sum_{R' \in {\rm FG}(k,s)} \II_{\R'}\cdot \hat{\II}_{\R'}\right] 
=\sum_{k\geq 1, s\geq 0}  \sum_{R' \in {\rm FG}(k,s)} \EE\left[\II_{\R'}\cdot \hat{\II}_{\R'} \right] $$
 $$= \sum_{k \geq 1, s\geq 0}  \sum_{R' \in {\rm FG}(k,s)} \EE[ \II_{\R'}]\cdot\EE[ \hat{\II}_{\R'}],
 \end{equation}
where the second equality is by linearity of expectation, and the third equality is by the independence assumption ($I_3$).
Given $\R'\in {\rm FG}(k,s)$,  let $C_1, C_2, \ldots, C_{s+c}$ be the random variables (ordered in any way)  that correspond to the catalysis probabilities of
the $s$ products of $\R'$ and the $c$ elements of $F\setminus C^{-}$. We can then write:

\begin{equation}\label{nice}
\EE[ \II_{\R'}] =\PP(\II_{R'}=1) = \EE[\PP(\II_{R'}=1|C_1, C_2, \ldots, C_{s+c})],
\end{equation}
where the second expectation is with respect to the random variables $C_i$.
The event $\II_{R'}=1$ occurs precisely when each of the $r$ reactions in $R'$ is catalysed by at least one of the $s+c$ elements in 
$(\pi(R')\setminus F) \cup (F\setminus C^{-})$.  By the independence assumption ($I_2$), 
\begin{equation}
\label{epr1}
\PP(\II_{R'}=1|C_1, C_2, \ldots, C_{s+c}) =  \prod_{r' \in R'} \left(1- \prod_{j=1}^{s+c} (1-C_j)\right) = \left(1- \prod_{j=1}^{s+c} (1-C_j)\right)^k.
\end{equation}
Set $V:= \prod_{j=1}^{s+c} (1-C_j)$.  \blue{Eqns.~(\ref{nice}) and (\ref{epr1}) then give:}
\begin{equation}
\label{epr2}
\blue{\EE[ \II_{\R'}]  = \EE[(1-V)^k] = \sum_{i=0}^k (-1)^i \binom{k}{i} \EE[V^i],}
\end{equation}
\blue{where the second equality is from the binomial expansion $(1-V)^k = \sum_{i=0}^k (-1)^i \binom{k}{i} V^i$, and linearity of expectation.}
Moreover,  for each $i\geq 0$, we have:
\begin{equation}
\label{epr3}
\EE[V^i] = \EE\left[ \left[\prod_{j=1}^{s+c} (1-C_j)\right]^i\right]=\EE\left[ \prod_{j=1}^{s+c} (1-C_j)^i\right] =\prod_{j=1}^{s+c} \EE[(1-C_j)^i]\\
 =\prod_{j=1}^{s+c} \lambda_i = \lambda_i^{s+c},
 \end{equation}
where the first two equalities are trivial algebraic identities, the third is by the independence assumption ($I_1$), the fourth is by definition and the last is trivial. 
\blue{Substituting Eqn.~(\ref{epr3}) into (\ref{epr2})} gives:
\begin{equation}
\label{epr4}
\EE[ \II_{\R'}] = \sum_{i=0}^k (-1)^i \binom{k}{i}\lambda_i^{s+c}.
\end{equation}

Turning to inhibition, a RAF subset $R'$ of $R$ in  ${\rm FG}(k,s)$ is a uRAF precisely if no reaction in $R'$ is blocked by any
of the $s+b$  elements of $(\pi(R')\setminus F) \cup (F\setminus B^{-})$.  By the independence assumption ($I'_2$), 
$$\PP(\hat{\II}_{R'}=1|B_1, B_2, \ldots, B_{s+b})  =  \prod_{r' \in R'}\left(\prod_{j=1}^{s+b} (1-B_j)\right)$$
$$= \left(\prod_{j=1}^{s+b} (1-B_j)\right)^k =\prod_{j=1}^{s+b} (1-B_j)^k. $$
Applying expectation (using the independence assumption ($I'_1$)), together with the identity $\EE[(1-B_j)^k] = \hat{\lambda}_k$ gives:
\begin{equation}
\label{epr5}
\EE[\hat{ \II}_{\R'}] =\hat{\lambda}_k^{s+b}.
\end{equation}

Combining \blue{Eqns.~(\ref{epr4}) and (\ref{epr5})} into Eqn.~(\ref{nicer}) gives the first equation in Part (i). The second is then obtained by putting $\hat{\lambda}_i = 1$ for all $i$.

 \bigskip

{\em Parts (ii) and (iii):} 
Observe that the function $u=(1-y)^k$ for $k \geq 1$  is convex and strictly convex when $k>1$. 
Thus, by Jensen's Inequality,  for any  random variable $Y$, we have:
\begin{equation}
\label{in}
\EE[(1-Y)^k] \geq (1-\EE[Y])^k,
\end{equation}
with a strict inequality when $Y$ is nondegenerate and $k>1$.
  
For Part (ii), let $\blue{V= } \prod_{j=1}^{s+c} (1-C_j)$. Then \blue{by the first equality in Eqn.~(\ref{epr2}) we have:}
$$\EE[ \II_{\R'}] = \EE[(1-V)^k],$$
\blue{and by Inequality~(\ref{in}) (with $Y=V$) we have:
\begin{equation}
\label{in2}
\EE[ \II_{\R'}] \geq (1-\EE[V])^k,
\end{equation}
\blue{and the inequality is strict when $V$ is nondegenerate and $k>1$. }
By the independence assumption $(I_1)$, and noting that $\EE[(1-C_j)] = 1-\mu_C$ we have:
\begin{equation}
\label{in3}
\EE[V] = \EE[ \prod_{j=1}^{s+c} (1-C_j)] = \prod_{j=1}^{s+c}\EE[(1-C_j)] = (1-\mu_C)^{s+c},
\end{equation}
and substituting Eqn.~(\ref{in3}) into Inequality~(\ref{in2}) gives:}
$$\EE[ \II_{\R'}]  \geq   (1-(1-\mu_C)^{s+c})^k,$$
with equality only for the uniform model. 
This gives Part (ii). 

\bigskip

For Part (iii)(a), Inequality (\ref{in}) implies that $\hat{\lambda}_k =\EE[(1-B)^k)] \geq (1-\mu_B)^k$.  
\blue{Let $H(k,s) := \left(\sum_{i=0}^k (-1)^i \binom{k}{i} \lambda_i^{s+c}\right)$.  By Eqn. (\ref{epr4}), $H(k,s) = \EE[ \II_{\R'}]$ for $\R' \in {\rm FG}(k,s)$ and so $H(k,s) \geq 0$. 
Thus, by  Eqn.~(\ref{mumu2}) we have:
$$\mu_{\rm uRAF}= \sum_{k\geq 1, s\geq 0}  n_{k,s}  \cdot H(k,s) \cdot  \hat{\lambda}_k^{s+b} \geq  \sum_{k\geq 1, s\geq 0}  n_{k,s} \cdot  H(k,s) \cdot (1-\mu_B)^{k(s+b)}, $$
and the right-hand side of this inequality is the value of $\mu_{\rm uRAF}$ for the uniform model of inhibition. }

\bigskip

For Part (iii)(b), 
suppose that $Y$ is a random variable taking values in $[0,1]$ with mean $\eta$ and let $Y_0$ be the random variable that
takes the value 1 with probability $\eta$ and $0$ otherwise. Then $\EE[Y_0^m] = \eta$ for all $m \geq 1$, and $\EE[Y^m] \leq \EE[Y^2] \leq \eta$ for all $m\geq 1$ (since $Y^m \leq Y^2 \leq Y$ because $Y$ takes values in $[0,1]$); moreover,
$\EE[Y^2]= \eta$ if and only if $\EE[Y(1-Y)] = 0$, which implies that $Y=Y_0$. 
Now apply this to $Y= (1-B)$ and $m=k$ to deduce  for the distributions on $B$ that have a given mean $\mu_B$,  $\hat{\lambda}_k$ is maximised when the distribution  takes the value $1$ with probability $\mu_B$ and
zero otherwise.
\hfill$\Box$

\section{Applications}

\subsection{Inhibition-catalysis trade-offs under the uniform model}

For any model in which catalysis and inhibition are uniform,  Theorem~\ref{thm1} provides a simple prediction concerning how the expected number of uRAFs compares with a model with zero inhibition (and a lower catalysis rate).   To simplify the statement, we will assume $b=c$ and we will write $\mu_{\rm uRAF}(p, tp)$ to denote the dependence of $\mu_{\rm uRAF}$ on 
$\mu_C=p$ and $\mu_B = tp$ for some value of $t$.  
We will also write $p = \nu /N$, where $N$ is the total number of molecule types that are in the food set or can be generated by a sequence of reactions in  $\R$.  We assume in the following result that $p$ is small (in particular, $< 1/2$) and $N$ is large (in particular, $(1-\nu/N)^N$ can be approximated by $e^{-\nu}$). 

The following result (which extends Theorem 2 from \cite{hor16}) applies to any chemical reaction system and provides a lower bound on the expected number of uRAFs in terms of the expected number of RAFs in the system with no inhibition (and half the catalysis rate); its proof relies on Theorem~\ref{thm1}. 
\blue{Roughly speaking, Corollary~\ref{thm2} states that  for any chemical reaction system with uniform catalysis, if one introduces a limited degree of inhibition then by doubling the original catalysis rate, the expected number of uninhibited RAFs is at least as large as the original number of expected RAF  before inhibition was present (and at the original catalysis rate). }

\begin{corollary}
\label{thm2}
For  all non-negative values of $t$ with $t \leq \frac{1}{\nu}\ln(1+e^{-\nu})$, the following inequality holds:
$$
\mu_{\rm uRAF}(2p, tp) \geq \mu_{\rm RAF}(p, 0).
$$
\end{corollary}

\begin{proof}
\blue{By Theorem~\ref{thm1}, and Remark (4) following this theorem,  and noting that $\mu_C =p$ and $\mu_B=tp$ we have:
\begin{equation}
\label{por2}
\mu_{\rm uRAF}(2p, tp)= \sum_{k \geq 1, s\geq 0}  n_{k,s} \left[(1- (1-2p)^{s+c})\cdot (1-tp)^{s+c}\right]^k,
\end{equation}
which can be re-written as:}
\begin{equation}
\label{por2plus}
\mu_{\rm uRAF}(2p, tp)= \sum_{k \geq 1, s\geq c}  n_{k,s-c} \left[(1- (1-2p)^{s})\cdot (1-tp)^{s}\right]^k.
\end{equation}
Thus (putting $t=0$ in this last equation) we obtain:
\begin{equation}
\label{por3}
\mu_{\rm RAF}(p, 0)= \sum_{k \geq 1, s\geq c}  n_{k,s-c} \left[1- (1-p)^{s}\right]^k.
\end{equation}
Now, for each $x\in (0, 0.5)$, we have: $$1-(1-2x)^s\geq 1-(1-x)^{2s} = (1-(1-x)^s)(1+(1-x)^s).$$
Thus (with $x=p$), we see that the term inside the square brackets in Eqn.~(\ref{por2plus}) exceeds the term in square brackets in Eqn.~(\ref{por3}) by a factor of
$(1+(1-p)^s)(1-tp)^s$, and this is minimised when $s = N$ (the largest possible value $s$ can take).  \blue{ Setting $s=N$ and writing $p = \nu/N$} we have 
$$(1+(1-p)^s)(1-tp)^s   \blue{ = (1+ (1-\nu/N)^N(1-t\nu/N)^N} \sim (1+e^{-\nu}) e^{-t\nu}$$ and the last term on the right is at least 1 when $t$ satisfies the stated inequality (namely, $t \leq \frac{1}{\nu}\ln(1+e^{-\nu})$).  \blue{Thus $(1+(1-p)^s)(1-tp)^s \geq 1$, for all $s$ between 1 and $N$ and so}
each term in Eqn.~(\ref{por2plus}) is greater or equal to the corresponding term in square brackets in Eqn.~(\ref{por3}), which justifies the inequality in Corollary~\ref{thm2}.
\end{proof}

%yyy
%%%UPTOHERE

\subsection{Explicit calculations for two models on a subclass of networks}
\label{relation}

For the remainder of this section, we consider  {\em elementary} chemical reaction systems (i.e. systems for which each reaction has all its reactants in the food set, as studied in \cite{ste}),  with the further conditions that: 
(i) each reaction has exactly one product,
(ii) different reactions produce different products,
(iii)  no reaction is inhibited, and 
(iv) no food element catalyses any reaction.

We can associate with each such system a directed graph $\G$ on the set $X-F$ of products of the reactions, with an arc from $x$ to $y$ if $x$ catalyses the reaction that produces $y$
(this models a setting investigated in \cite{j1, j2}). 
RAF subsets are then in one-to-one correspondence with the subgraphs of $\G$ for which each vertex has indegree at least one. In particular, a RAF exists if and only if there is a directed cycle in $\G$ (which could be an arc from a vertex to itself).\footnote{An asymptotic study of the emergence of first cycles in large random directed graphs was explored in \cite{bol}.}  In this simple set-up, if $N$ denotes the number of reactions (= number of non-food molecule types) then:
$$n_{k,s} = \begin{cases}
\binom{N}{k}, & \mbox{ if $k=s$;}\\
0, & \mbox{ otherwise.}
\end{cases}
$$ 
Applying Theorem~\ref{thm1}(i) gives:
\begin{equation}
\label{ab1}
\mu_{\rm RAF}  = \sum_{j=1}^N \binom{N}{j} \left ( \sum_{i=0}^j (-1)^i\binom{j}{i} \lambda_i^j\right).
\end{equation}

Regarding catalysis, consider first the {\bf all-or-nothing model}, for which $\lambda_i= 1-\pi=1-\mu_C$ for $i\geq 1$ (and $\lambda_0=1$).  
 Eqn.~(\ref{ab1}) simplifies  to:
\begin{equation}
\label{ab1a}
\mu_{\rm RAF} = 2^N - (2-\mu_C)^N,
\end{equation}
and we provide a proof of this in the Appendix. 

This expression can also be derived by the following direct argument. First, note that a subset $S$ of the $N$ products of reactions does not correspond to a RAF if and only if each of the $|S|$ elements $x$ in $S$ has $C_x=0$. 
The random variable $W=|\{x: C_x =1\}|$  follows the binomial distribution $Bin(N, \mu_C)$, and the  proportion of sets of size $N$ that avoid a given set $S$ of size $m$
is $2^{-m}$.  Thus the expected proportion of subsets that are not RAFs is the expected value of $2^{-W}$ where $W$ is the binomial distribution above. Applying standard combinatorial identities then leads to Eqn.~(\ref{ab1a}).

The probability of a RAF for the all-or-nothing models is also easily computed:
\begin{equation}
\label{ab2a}
P_{\rm RAF}  = 1-(1-\mu_C)^N.
\end{equation}
Notice that one can select  $\mu_C$ to tend to 0  in such a way  $P_{\rm RAF}$ converges to 0 exponentially quickly with $N$ while $\mu_{\rm RAF}$ tends to infinity at an exponential rate with $N$ (this requires $\mu_C$ to decay sufficiently fast with $N$ but not too fast, e.g. $\mu_C = \Theta(N^{-1-\delta})$ for $\delta>0$). 
Comparing Eqns.~(\ref{ab1a}) and (\ref{ab2a}), we also observe the following identity: $$\mu_{\rm RAF}(\mu_C) = 2^N P_{\rm RAF}(\mu_C/2 ).$$

By contrast, for the {\bf uniform model},  applying straightforward algebra to Eqn.~(\ref{ab1}) leads to
\begin{equation}
\label{ab3x}
\mu_{\rm RAF}  = \sum_{j=1}^N \binom{N}{j} \left(1- (1-\mu_C)^j\right)^j.
\end{equation}

\blue{ We now use these formulae to investigate the relationship between $P_{\rm RAF}$ and $\mu_{\rm RAF}$ in elementary chemical reaction systems (satisfying conditions (i)--(iv))  as $N$ becomes large; in particular the impact of the choice of model (all-or-nothing vs uniform) on this relationship. }

\bigskip

\noindent {\bf Asymptotic properties of the two models at the catalysis level where RAFs arise:}    For  the all-or-nothing and uniform models, RAFs arise with a given (positive) probability, provided that $\mu_C$ converges to 0 no faster than  $N^{-1}$
as $N$ grows.  Thus, it is helpful to write $\mu_C = \gamma/N$ to compare their behaviour as $N$ grows.

For the all-or-nothing model, Eqns.~(\ref{ab1a}) and (\ref{ab2a})  reveal that:
$$\frac{\mu_{\rm RAF}}{ P_{\rm RAF}} = 
2^N \frac{\left(1-\left(1-\frac{\gamma}{2N}\right)^N\right)}{\left(1-\left(1-\frac{\gamma}{N}\right)^N\right)}
\sim 2^N \left(\frac{1-\exp(-\gamma/2)}{1-\exp(-\gamma)}\right),$$
where $\sim$ is asymptotic equivalence as $N$ becomes large (with $\gamma$ being fixed),
and so: 
\begin{equation}
\label{abu}
\frac{\mu_{\rm RAF}}{ P_{\rm RAF}} \sim 2^{N-1}(1 + O(\gamma)),
\end{equation}

Let us compare this with the uniform model with the same $\mu_C$ (and hence $\gamma$) value.
It can be shown that when $\gamma< e^{-1}$, we have:
\begin{equation}
\label{ab0}
\lim_{N \rightarrow \infty}  \sum_{j=1}^N \binom{N}{j} \left(1- (1-\gamma/N)^j\right)^j =   \gamma + o(\gamma).
\end{equation}
where $o(\gamma)$ has  order $\gamma^2$ as $\gamma \rightarrow 0$ (a proof is provided in  the Appendix).

By Theorem 1 of  \cite{hor2} (and for any value of $N$ and assuming $\gamma<1$), we have:
\begin{equation}
\label{ab3y}
1-\exp(-\gamma) \leq P_{\rm RAF}  \leq  -\ln(1-\gamma).
\end{equation}
In particular, for small $\gamma$ and the uniform model we have: 
\begin{equation}
\label{ab4y}
P_{\rm RAF} = \gamma + o(\gamma). 
\end{equation}
 Eqns.~(\ref{ab1}), (\ref{ab0}), and  (\ref{ab4y}) provide the following result for the  uniform model when $\gamma < e^{-1}$:  
\begin{equation}
\label{abu2}
\frac{\mu_{\rm RAF}}{ P_{\rm RAF}} \sim 1  + O(\gamma),
\end{equation}
where $\sim$ again denotes asymptotic equivalence as $N$ becomes large (with $\gamma$ fixed). 

Comparing  Eqns.~(\ref{abu}) and (\ref{abu2}) reveals a key difference in the ratio $\mu_{\rm RAF}/ P_{\rm RAF}$ between the all-or-nothing and uniform models when $N$ is large and $\gamma$ is small: the former equation involves an exponential term in $N$, while the second does not.   This can be explained as follows.   In the all-or-nothing model, the existence of a RAF comes down to whether or not there is a reaction $r$  that generates a universal catalyst; when there is, then any subset of the $N$ reactions that contains $r$ is a RAF.  By contrast, with the uniform model at a low catalysis level where RAF are improbable, if a RAF exists, there is likely to be only one.  \blue{Note that the results in this section are particular to  chemical reaction systems that are elementary and  satisfy properties (i)--(iv) as described at the start of this section.}

\section{Concluding comments}
In this paper, we have focused on the expected number of RAFs and uRAFs (rather than the probability of at least one such set existing),  as this quantity can be described explicitly, and generic results described via  this expression can be derived  (e.g. in Parts (ii) and (iii) of Theorem~\ref{thm1} and Corollary~\ref{thm2}).  Even so, the expressions in Theorem~\ref{thm1} involve quantities $n_{k,s}$ that may be difficult to quantify exactly; thus in the second part of the paper, we consider more restrictive types of systems.

In our analysis, we have treated inhibition and catalysis as simple and separate processes. However, a more general approach would allow reactions to proceed under rules that are encoded by Boolean expressions. For example, the expression $(a \wedge b) \vee c  \vee (d \wedge \neg e)$  assigned to a reaction $r$ would allow $r$ to proceed if at least one of the following holds: (i) both $a$ and $b$ are present as catalysts, or (ii) $c$ is present as a catalyst or (iii) $d$ is present as a catalyst and $e$ is not present as an inhibitor. Extending the results in this paper to this more general setting could be an interesting exercise for future work.

\section{Acknowledgements}
\blue{We thank the two reviewers for a number of helpful comments on an earlier version of this manuscript.}

\bigskip

\section{Appendix: Justification of Eqns.~(\ref{ab1a}) and (\ref{ab0}).}

{\em Eqn.~(\ref{ab1a}):}
We use three applications of the standard binomial identity  $\sum_{k=0}^n \binom{n}{k}x^k = (1+x)^n$. 
Set $\lambda = 1-\mu_C$.  Since $\lambda_i = \lambda$ for $i\geq 1$, the binomial identity (with $x=-1, n=j, k= i$) gives:
$$\sum_{i=1}^j (-1)^i \binom{j}{i} \lambda_i^j  = \lambda^j \cdot \sum_{i=1}^j (-1)^i\binom{j}{i} = \lambda^j \cdot \left( \sum_{i=0}^j  \binom{j}{i}(-1)^i - 1\right) =  \lambda^j (0^j-1)= -\lambda^j,$$
for each $j\geq 1$. 
Thus, adding in the additional term (for $i=0$ where $\lambda_0 =1$), we obtain:
$\sum_{i=0}^j (-1)^i \binom{j}{i} \lambda_i^j   =1-\lambda^j.$
 Eqn.~(\ref{ab1}) now gives:
$$\mu_{\rm RAF} = \sum_{j=1}^N \binom{N}{j} (1-\lambda^j) = \sum_{j=0}^N \binom{N}{j} (1-\lambda^j) = 2^N -(1+\lambda)^N,$$
where the third equality involves two further applications of the binomial identity (with $n=N, k=j$, and with one application using $x=1$, the other using $x=\lambda$). \hfill$\Box$

\bigskip

{\em Eqn.~(\ref{ab0}):}    Observe that the $j$-th term on the LHS of Eqn.~(\ref{ab0}) is
$ \binom{N}{j} \left(1- (1-\gamma/N)^j\right)^j$. For $j=1$ this simplifies to $\gamma$.
A simple proof by induction shows that for all $j\geq 1$, and all $x \in (0,1)$, we have: $(1-x)^j \geq 1-xj$ and so $(1-(1-x)^j)^j \leq (xj)^j$.  
Applying this with $x=\gamma/N\in (0,1)$, 
 the LHS of Eqn.~(\ref{ab0}) is bounded below by $\gamma$ (the term where $j=1$) and is bounded above by:
 $$\gamma +  \sum_{j \geq 2} \binom{N}{j} \frac{\gamma^{j}j^j}{N^j} \leq \gamma + \sum_{j \geq 2} \frac{\gamma^{j}j^j}{j!},$$
 where the inequality follows from $\binom{N}{j}/N^j \leq \frac{1}{j!}$.  Next, observe that, by Stirling's formula for $j!$, we have:  $$ \gamma^{j} \cdot \frac{j^j}{j!} \sim \frac{(\gamma e)^j}{\sqrt{2\pi j}},$$
and this term converges to zero at exponential rate  as $j$ increases provided that $\gamma< e^{-1}$; in particular, for $\gamma < e^{-1}$, the sum
$\sum_{j \geq 2} \frac{\gamma^{j}j^j}{j!}$ converges to a constant of order $\gamma^2$ as $\gamma \rightarrow 0$. 
\hfill$\Box$
\end{document}